\documentclass[a4paper,onecolumn,11pt,accepted=2026-02-02]{quantumarticle}
\pdfoutput=1
\usepackage[utf8]{inputenc}
\usepackage[english]{babel}
\usepackage[T1]{fontenc}
\usepackage{amsmath}
\usepackage{amssymb} 
\usepackage{amsthm}
\usepackage{amsfonts}
\usepackage{hyperref}
\usepackage{booktabs}          
\usepackage{tcolorbox}         
\usepackage{multirow}
\usepackage{tabularx}
\usepackage[ruled, vlined, linesnumbered,nofillcomment]{algorithm2e}

\usepackage{tikz}

\usepackage{lipsum}

\usepackage[numbers]{natbib}

\def\a{\boldsymbol \alpha}

\def\M{\mathcal{M}}
\def\R{\mathcal{R}}

\def\L{\mathcal{L}}

\def\1{\mathbbm{1}}

\def\eps{\varepsilon}

\DeclareUnicodeCharacter{00A0}{ }

\newtheorem{theo}{Theorem}

\newtheorem{lemma}[theo]{Lemma}

\begin{document}

\title{Bosonic quantum Fourier codes}

\author{Anthony Leverrier}
\affiliation{Inria Paris, France}
\orcid{0000-0002-6707-1458}
\email{anthony.leverrier@inria.fr}
\thanks{This work was funded by the Plan France 2030 through the project ANR-22-PETQ-0006.}

\maketitle

\begin{abstract}
While 2-level systems, aka qubits, are a natural choice to perform a logical quantum computation, the situation is less clear at the physical level. Encoding information in higher-dimensional physical systems can indeed provide a first level of redundancy and error correction that simplifies the overall fault-tolerant architecture. A challenge then is to ensure universal control over the encoded qubits. Here, we explore an approach where information is encoded in an irreducible representation of a finite subgroup of $U(2)$ through an inverse quantum Fourier transform. We illustrate this idea by applying it to the real Pauli group $\langle X, Z\rangle$ in the bosonic setting. The resulting \textit{two-mode Fourier cat code} displays good error correction properties and admits an experimentally-friendly universal gate set that we discuss in detail.  
\end{abstract}


\section{Introduction}
\label{sec:intro}

Years of theoretical advances on quantum error correction and fault tolerance have uncovered a fundamental tension between information processing and error correction in quantum encoding schemes: information well protected against relevant sources of imperfection tends to be difficult to manipulate. The difficulty to combine error correction capabilities and a fault-tolerant manipulation of the logical information is exemplified by the Eastin-Knill theorem: logical gates of a quantum code admitting a transversal implementation cannot form a universal gate set~\cite{EK09}. 
Infinite-dimensional systems such as bosonic modes are a natural playground to try to evade this restriction and form a promising avenue towards reducing the cost of fault tolerance. For instance, the cat code~\cite{MLA14} offers an intrinsic protection against bit flips and one can then turn to classical error correction to fight phase errors~\cite{GM19}. More generally, an appealing architecture for fault-tolerant quantum computing would rely on a concatenation of a bosonic encoding with a quantum LDPC code~\cite{VAW19,CNA22}. Pushing the idea even further, one could even envision a purely bosonic encoding without any concatenation, \textit{e.g.}~with tiger codes~\cite{XWV24}.  
Recent experimental progress towards demonstrating cat~\cite{GFP20,RBG24,PNH25} or GKP~\cite{FNM19,CET20,SER23} encodings is particularly encouraging and suggests to look for potential alternative bosonic codes that would combine a good resistance to loss with a universal gate set.

There exist many possible bosonic encodings of a single qubit: the dual-rail encoding~\cite{TWB23} involves a single photon in two bosonic modes $|1\rangle |0\rangle$ and $|0\rangle |1\rangle$, the cat qubit exploits superpositions of two simple coherent states $|\alpha\rangle \pm |-\alpha\rangle$, the GKP code~\cite{GKP01} relies on more complex grid states and offers better protection at the price of a more complex implementation. Many other candidates exist, ranging from the pair-cat code~\cite{AMG19} to binomial codes~\cite{MSB16}, rotation symmetric codes~\cite{GCB20} and multimode spherical codes~\cite{JIB23}. 
In general, the approach to design a new code is to focus first on protecting against the most relevant sources of imperfections (typically loss in the bosonic setting) and then understanding how to perform logical gates in a fault-tolerant fashion. By contrast, one could equally well focus on encodings where some gates are naturally fault tolerant and only then optimize their error correction capabilities. This strategy was pioneered by Gross for spin encodings~\cite{gro21} and also explored by Kubischta and Teixeira for multiqubit codes~\cite{KT23,KT24} and Denys and Leverrier for bosonic codes~\cite{DL24}. The idea behind this last work is that starting with a finite group $G \subset U(2)$ and a ``nice'' physical representation $\pi$ of that group on the relevant physical space, it is possible to systematically obtain all the \textit{covariant encodings}, \textit{i.e.}~those with the property that applying $\pi(g)$ on the physical state implements the gate $g\in G$ at the logical level. The goal of the present paper is to explore in more detail a specific instance of such a code, where the encoding unitary is the inverse quantum Fourier transform of the group. Interestingly, it turns out that it encodes 2 qubits with very distinct properties. We think of the first one as our logical qubit, for which we want universal control, and of the second one as a gauge or auxiliary qubit useful to implement gates on the first qubit.

We formalize this quantum Fourier encoding in Section \ref{sec:qfc} before focusing in detail on the simplest possible nontrivial instance of the construction,  based on the Pauli group $G = \langle X,Z\rangle$ with the natural passive Gaussian unitary representation $\pi$ on two bosonic modes: the gate $\pi(X)$ simply swaps the two modes and $\pi(Z) = (-1)^{\hat{n}_2}$ applies a $\pi$ phase on the second mode. We present a universal gate set for the logical qubit in Section \ref{sec:univ}. In particular, the implementation of the Hadamard gate crucially relies on a code deformation technique involving the second encoded qubit. We discuss the error correction properties of the encoding in Section \ref{sec:ec}.

\section{The framework and comparison with other bosonic codes}
\label{sec:qfc}

\subsection{Quantum Fourier codes}
Consider a finite subgroup $G \subset U(d)$ and its associated quantum Fourier transform~\cite{CvD10}:
\begin{align}\label{eqn:ft}
F_G := \sum_{g\in G} \sum_{\rho\in \hat{G}} \sqrt{\frac{d_\rho}{|G|}} \sum_{\ell,m=1}^{d_\rho} \langle \ell|\rho(g)|m\rangle \: |\rho, \ell,m\rangle\langle g|,
\end{align}
where $\hat{G}$ is the set of inequivalent irreducible representations of $G$, and the states $|g\rangle$ labeled by group elements are assumed to be orthonormal. Here, $\rho$ is an irreducible representation of the group of dimension $d_\rho$ and the definition of $F_G$ depends on a choice of orthonormal basis $\{ |\rho,\ell,m\rangle  \: : \: 1 \leq \ell,m \leq d_\rho\} \subset \mathbbm{C}^{d_\rho} \otimes \mathbbm{C}^{d_\rho}$ for each irreducible representation of dimension greater than 1. 
If we define the left and right regular representations $\L$ and $\R$ of the group as:
\begin{align}\label{eqn:reg}
\L(g) |h\rangle = |gh\rangle,\qquad \R(g)|h\rangle = |hg^{-1}\rangle,
\end{align}
then the quantum Fourier transform is the unitary that simultaneously diagonalizes both the left and right regular transformations: for all $g \in G$,
\begin{align}\label{eqn:block}
F_G \L(g)F_G^\dag = \bigoplus_{\rho \in \hat{G}} (\rho(g) \otimes \mathbbm{1}_{d_\rho}), \qquad 
F_G \R(g)F_G^\dag = \bigoplus_{\rho \in \hat{G}} (\mathbbm{1}_{d_\rho} \otimes \rho(g)).
\end{align} 
In the following, we will choose the $d$-dimensional defining representation of $U(d)$, denoted by $\lambda$, and will often write $g$ instead of $\lambda(g)$ when there is no ambiguity. We denote by $L$ and $M$ the two $d$-dimensional factors of the space spanned by $|\lambda, \ell,m\rangle$. These two registers will be our logical and multiplicity registers, respectively. 

with uni
When the left regular representation is easy to implement for some system, it makes sense to try to encode information in the logical register $L$, and to use the multiplicity register $M$ as some gauge register.
We therefore define an encoding map $\mathcal{E} :  L \otimes M \to \mathrm{Span}(|g\rangle \: : \: g\in G)$:
\begin{align}\label{eqn:enc}
\mathcal{E} : |\ell\rangle \otimes |m\rangle \quad \mapsto \quad & F_G^\dag |\lambda, \ell, m\rangle=  \sqrt{\frac{d}{|G|}} \sum_{g \in G}  \langle m |\lambda(g)^\dag | \ell\rangle |g\rangle.
\end{align}
We abuse notation and also write $L$ and $M$ for their image by $\mathcal{E}$. 
By construction, this encoding guarantees that the physical gate $\L(g)$ implements a logical operation $\lambda(g)$ on the logical register:
\begin{align}\label{eqn:main-identity}
\L(g) \mathcal{E}(|\ell\rangle|m\rangle) = \L(g)F_G^\dag |\lambda, \ell, m\rangle = F_G^\dag (F_G \L(g)F_G^\dag) |\lambda, \ell, m\rangle
= \mathcal{E} ( \lambda(g) |\ell\rangle \otimes |m\rangle).
\end{align}

It may not be obvious at first sight that there are natural settings where the states $|g\rangle$ are easy to prepare and the operations $\L(g)$ are easy to implement. But this is the case and bosonic codes provide a particularly nice illustration.
For instance, if the group $G$ admits a $d$-dimensional unitary representation $\lambda$, one can consider an arbitrary $d$-mode coherent state $|\a\rangle  = |\alpha_1, \ldots, \alpha_d\rangle$ such that the vectors $\lambda(g)\a$ are all distinct, and the physical representation $\pi$ of $G$ corresponding to passive Gaussian unitaries defined by 
\begin{align}\label{eqn:pi}
\pi(g) |\a\rangle = |\lambda(g) \a\rangle, \quad \forall g \in G.
\end{align}
Here, $\pi(g)$ is a unitary Gaussian operator of the $d$-mode bosonic Hilbert space.
It is well known that the coherent states $|g\a\rangle$ (we omit $\lambda$ for conciseness) are not orthogonal, but they are linearly independent when the group $G$ is finite. One can readily obtain an orthonormal basis labeled by the group elements:
\begin{align}
|g\rangle := \sum_{h \in G} [\Gamma^{-1/2}]_{h,g} |h\a\rangle,
\end{align}
where $\Gamma$ is the Gram matrix of the family $\{|g\a\rangle\}$, namely $[\Gamma]_{g,h} = \langle g\a|h\a\rangle$. To see this, 
\begin{align*}
\langle g|h\rangle &= \sum_{k,\ell} [\Gamma^{-1/2} ]_{gk} [\Gamma^{-1/2}]_{\ell h} \langle k\a|\ell \a\rangle = \sum_{k,\ell} [\Gamma^{-1/2} ]_{gk} [\Gamma^{-1/2}]_{\ell h}\Gamma_{k\ell}= \delta_{g,h}. 
\end{align*}
The representation $\pi$ acts like the left regular representation on these states:
\begin{align}\label{eqn:pi=E}
\pi(g) |h\rangle &= \sum_{k \in G} [\Gamma^{-1/2}]_{k,h}  |gk\a\rangle=  \sum_{k \in G} [\Gamma^{-1/2}]_{gk,gh}  |gk\a\rangle
 = |gh\rangle = \L(g) |h\rangle
\end{align}
where the second equality follows from the fact that the Gram matrix and its powers commute with $\pi(g)$.

We arrive at the definition of the quantum Fourier encoding $\mathcal{E}(|\ell\rangle|m\rangle)$ of $|\ell\rangle|m\rangle \in L\otimes M$, that we also denote by $|\widehat{\ell,m}\rangle$:
\begin{align}
|\widehat{\ell,m}\rangle := \sum_{g \in G} [\Gamma^{-1/2} F_G^\dag]_{g,\lambda \ell m}  |g\a\rangle.
\end{align}
By \eqref{eqn:enc} and \eqref{eqn:pi=E}, we know that $\pi$ acts similarly to the left regular representation for this state, and \eqref{eqn:main-identity} implies that for any group element, $\pi(g)$ implements the logical gate $\lambda(g)$ on the logical register $L$.
It is tempting to consider a large subgroup of $U(d)$ because the corresponding logical gates are easy to implement by construction. However, this leads to more complicated code states that might be delicate to prepare.

Note that the quantum Fourier code is very similar to the covariant encoding of \cite{DL24} which applies the inverse quantum Fourier transform directly to the coherent states $|g\a\rangle$ instead of the orthonormal states $|g\rangle$, and picks an arbitrary state $|\Omega \rangle \in M$. 
Up to a global normalization, the encoded qudit in \cite{DL24} is given by
\begin{align}\label{eqn:dl24}
|\overline{\ell}\rangle \propto \sum_{g\in G} [F^\dag]_{g,\lambda \ell \Omega} |g\a\rangle.
\end{align}
A feature of the Fourier code that we consider here is that it also stores information in the multiplicity space. This extra degree of freedom will be useful to design a universal set of operations for the logical qudit.

\subsection{The two-mode Fourier cat code}

The relevant groups for applying the Fourier encodings are the noncommutative ones.
This is because all the irreducible representations of an abelian group are one-dimensional, preventing one from encoding a nontrivial logical system in a single irreducible representation. For instance, rotation symmetric bosonic codes such as the cat qudit encoding encode the logical basis codewords in different irreducible representations of the group $\mathbbm{Z}_N$, as detailed in Appendix \ref{sec:cat}. 
By contrast, we focus here on a noncommutative group $G$ and the simplest example may be the real Pauli group $G = \langle X,Z\rangle = \{ \pm \1, \pm X, \pm Z, \pm XZ\} \subset U(2)$. It admits 5 inequivalent irreducible representations: four of them are one-dimensional and the standard defining representation $\lambda$ is 2-dimensional: $\lambda(X) = \left[\begin{smallmatrix} 0 & 1 \\ 1 & 0 \end{smallmatrix}\right], \lambda(Z) = \left[\begin{smallmatrix} 1 & 0 \\ 0 & -1 \end{smallmatrix}\right]$.

One could pick any 2-mode coherent state to instantiate the construction. It turns out, however, that not all coherent states are born equal and we will choose $|\a\rangle = |\alpha, i\alpha\rangle$ with $\alpha = \sqrt{\pi/2}$ in the following. 
This ensures that the even cat states $|0_\alpha\rangle \propto |\alpha\rangle + |-\alpha\rangle$ and $|0_{i\alpha}\rangle \propto |i\alpha\rangle + |-i\alpha\rangle$ and odd cat states $|1_{\alpha}\rangle \propto |\alpha\rangle - |-\alpha\rangle$ and $|1_{i\alpha}\rangle \propto |i \alpha\rangle - |-i\alpha\rangle$ are orthogonal. Moreover, for this value, the Gram matrix $\Gamma$ of the family $\{|g\a\rangle \: : \: g\in G\}$ is diagonalized by the Fourier transform of the group, and the restriction of $F \Gamma F^\dag$ to the irreducible representation $\lambda$ is a scalar. 
As a consequence, the four codestates $|\widehat{\ell,m}\rangle$ are products of single-mode cat states:
\begin{align}\label{eqn:basis}
|\widehat{0,0}\rangle = |1_\alpha\rangle |0_{i\alpha}\rangle, \quad |\widehat{0,1}\rangle = |1_{i\alpha}\rangle |0_{\alpha}\rangle, \quad |\widehat{1,0}\rangle =  |0_{i\alpha}\rangle|1_\alpha\rangle, \quad |\widehat{1,1}\rangle =  |0_{\alpha}\rangle|1_{i\alpha}\rangle.
\end{align}
In particular, when the state of the multiplicity qubit is $|0\rangle$, the logical code states $|\widehat{\ell,0}\rangle$ coincide with those of the covariant encoding of \cite{DL24} given in \eqref{eqn:dl24}.
By construction, the SWAP operator $\pi(X)$ that exchanges both modes performs a logical $X_{L}$ on the first qubit, while the phase-flip $(-1)^{\hat{n}_2} = \pi(Z)$ performs a logical $Z_{L}$ since it leaves even cat states invariant and adds a phase $-1$ to odd cat states. 
As already hinted at, even if this formally encodes 2 qubits, the idea is to promote the first one to be the true logical qubit while the second qubit (corresponding to the multiplicity space $M \cong \mathbbm{C}^2)$ will serve as an extra-register useful to implement specific gates.

From the expression of the logical codestates in \eqref{eqn:basis}, one can infer that the codespace manifold lies in the kernel of the Lindblad operators $\hat{a}_1^4 - \alpha^4$ and $\hat{a}_1^2 \hat{a}_2^2 + \alpha^4$, or alternatively by $\hat{a}_1^4 - \alpha^4$ and $\hat{a}_1^2 +\hat{a}_2^2$, and is stabilized by $(-1)^{\hat{n}_1 + \hat{n}_2 +1}$ enforcing that the total photon number parity is odd. These operators are similar to those of various alternative cat encodings.  They are summarized in Table \ref{table:stab}.

\begin{table}[htbp]
\centering
\begin{tcolorbox}[colframe=blue!70!black, title=\textbf{Stabilizers, Lindblad operators, logical states of cat-type bosonic codes}, width=1.0\textwidth]
\centering  
\footnotesize  
\setlength{\tabcolsep}{4pt}  
\renewcommand{\arraystretch}{1.2} 
\begin{tabular}{>{\centering\arraybackslash}p{1.7cm} | >{\centering\arraybackslash}p{1.8cm} >{\centering\arraybackslash}p{2.2cm} >{\centering\arraybackslash}p{2.6cm} >{\centering\arraybackslash}p{4.5cm}}
\toprule
& \textbf{4-legged cat} & \textbf{2-repetition cat} & \textbf{pair-cat} \cite{AMG19} & \textbf{2-mode Fourier} \\
\midrule
\textbf{Stabilizers} 
& $\hat{n} \mod 2$
& $\hat{n}_1 - \hat{n}_2 \mod 2$
& $\hat{n}_1 - \hat{n}_2$
& $\hat{n}_1 - \hat{n}_2 \mod 2$ \\
\midrule
\textbf{Lindblad operators}
& $\hat{a}^4 - \alpha^4$
& \begin{tabular}{c} $\hat{a}_1^2 - \alpha^2$ \\ $\hat{a}_2^2 - \alpha^2$ \end{tabular}
& $\hat{a}_1^2 \hat{a}_2^2 - \alpha^4$
& \begin{tabular}{c} $\hat{a}_1^4 - \alpha^4$ \\ $\hat{a}_1^2 \hat{a}_2^2 + \alpha^4$ \end{tabular} \\
\midrule
\textbf{$|0\rangle_L$} 
& $|0_\alpha\rangle$
& $|0_\alpha\rangle\,|0_\alpha\rangle$
& $\displaystyle \int_0^\pi |0_{\alpha e^{i\theta}}\rangle\,|0_{\alpha e^{-i\theta}}\rangle\,d\theta$ 
& \begin{tabular}{c} $|1_\alpha\rangle\,|0_{i\alpha}\rangle,$  $ |1_{i\alpha}\rangle\,|0_\alpha\rangle$ \end{tabular} \\
\textbf{$|1\rangle_L$} 
& $|0_{i\alpha}\rangle$
& $|1_\alpha\rangle\,|1_\alpha\rangle$
& $\displaystyle \int_0^\pi |1_{\alpha e^{i\theta}}\rangle\,|1_{\alpha e^{-i\theta}}\rangle\,d\theta$ 
& \begin{tabular}{c} $|0_{i\alpha}\rangle\,|1_\alpha\rangle,$  $ |0_\alpha\rangle\,|1_{i\alpha}\rangle$ \end{tabular} \\
\bottomrule
\end{tabular}
\end{tcolorbox}
\caption{For the 4-legged cat qubit and the two-mode Fourier code, the expressions of the logical states hold for the specific case with $\alpha= \sqrt{\frac{\pi}{2}}$. Note that the two logical states of the 4-legged cat are also first order approximations of the logical GKP states (keeping only 2 states in the coherent state basis expansion). The Lindblad operators of the two-mode Fourier code combine those of the 4-legged cat qubit and the pair-cat.}
\label{table:stab}
\end{table}

Before analyzing in more detail the two-mode Fourier cat code, it is insightful to compare it with other bosonic codes. 
Among single-mode encodings, \textit{rotation-symmetric bosonic codes}~\cite{GCB20} (including the standard cat qubit) are closely related and encode information in several irreducible representations of the abelian group $\mathbbm{Z}_N$. They share similar implementations with the two-mode Fourier code for the logical gates $S$ and $CZ$, relying respectively on self-Kerr and cross-Kerr interaction. Moreover, the Fourier code can correct a single-photon loss, exactly as  the 4-legged cat qubit~\cite{OPH16}. A major distinction is that the cat qubit is heavily biased, which is not the case of the Fourier code. 
At the other end of the spectrum, the \textit{GKP code}~\cite{GKP01,GP21} offers optimal performance against loss, and has the remarkable feature that all the logical Clifford gates can be implemented with Gaussian unitaries. The price to pay is the complexity of the state preparation and stabilization~\cite{SSL25}. 

It is also possible to encode a qubit in two bosonic modes. The \textit{dual-rail} encoding does that with a single-photon. In that case, the representation $\pi$ of the group $U(2)$ on the 2-mode Fock space restricts to the defining representation $\lambda$ on the space spanned by single-photon states. All logical single-qubit gates can therefore be implemented with passive Gaussian unitaries. A drawback of this approach is that it cannot correct a single-photon loss, only detect it. The two-mode binomial code generalizes the dual-rail encoding to improve loss tolerance, with more complicated code states~\cite{CLY97}.
Recently, the \textit{pair-cat code} was also introduced~\cite{AMG19} and has the advantage that measuring the loss error syndrome can be done without turning off the dissipation (which shares a Lindblad operator with the Fourier code: see Table \ref{table:stab}). When there's no gain in the channel, the pair-cat code can also correct a single-photon loss on an arbitrary mode. 
An alternative two-mode code is to concatenate a cat qubit with a repetition code of length 2. This \textit{2-repetition cat} encoding can detect a single-photon loss. 
Finally, the two-mode Fourier encoding is closely related to the covariant code of \cite{DL24} which only encodes a single qubit. As detailed in Section \ref{sec:univ}, the 2-qubit code space is instrumental to the design of the logical Hadamard gate for the two-mode Fourier code.

\section{A universal gate set}
\label{sec:univ}

By construction, the gates from the group $G$ are easy to obtain through the physical representation $\pi$. When $G = \langle X,Z\rangle$, this only gives Pauli operators, however, which are insufficient for universal quantum computing. One also needs the ability to prepare some states, perform measurements, as well as Clifford and non-Clifford gates. 
For our specific choice of initial state $|\alpha\rangle |i\alpha\rangle$, the computational basis states are products of two single-mode cat states, which are now routinely prepared in the lab. Measuring the logical state in the $Z$ basis can be done by performing a photon parity measurement since $|0\rangle_L$ has an odd number of photons in the first mode and even in the second mode, while the opposite holds for the state $|1\rangle_L$. Measuring the operator $Y_M$ is also possible by measuring photon numbers modulo 4. Kerr interactions yield some Clifford gates, namely the phase gate $S$ and the $CZ$ gate, similarly to what is done for rotation-symmetric codes. 
The implementation of the Hadamard gate is more original. Applying $\pi(H)$ is a code deformation: the gate $H_L H_M$ is applied to both qubits, but the code is deformed to an equivalent code (a two-mode Fourier cat code with initial state $|e^{i\pi/4} \alpha\rangle |e^{-i\pi/4}\alpha\rangle$ instead of $|\alpha\rangle |i\alpha\rangle$). One can alternate between this gate and a phase gate to obtain a Hadamard gate $H_L$ applied only to the logical qubit. 
Finally, the gate $e^{i \theta Z_L Z_M}$ is obtained thanks to the quantum Zeno effect~\cite{MLA14} with a quadratic Hamiltonian drive $\hat{a}^2 + \hat{a}^{\dag 2}$. Provided that the state of the multiplicity qubit is fixed to $|0\rangle$ or $|1\rangle$, this yields a gate $e^{i \theta Z_L}$, thereby completing a universal gate set for the logical qubit.
These possible implementations are detailed in the remainder of this section, and Table \ref{table:op} provides a summary.

\begin{table}[htbp]
\centering
\begin{tcolorbox}[colframe=blue!70!black, title=\textbf{Physical implementation of logical operations}, width=1.0\textwidth]
\centering  
\footnotesize  
\setlength{\tabcolsep}{4pt}  
\renewcommand{\arraystretch}{1.2}  
\begin{tabular}{>{\arraybackslash}p{4.8cm} | >{\arraybackslash}p{8.7cm}}
\toprule
\textbf{logical operations} & \textbf{possible physical implementations} \\
\midrule
state preparation: $\mathcal{P}_{|\ell\rangle |m\rangle} $
& prepare two single-mode cat states\\
\midrule
logical measurement: $\M_{Z_L} $
& photon parity measurement in either mode\\
2-qubit measurement: $\M_{Z_L Y_M} $
& photon measurement modulo 4 in both modes\\
\midrule
logical Pauli gates: $X_L$, $Z_L$
& SWAP, $ (-1)^{\hat{n}_2}$\\
$X_M Z_M$ &
$-i^{\hat{n}_1 + \hat{n}_2}$\\
\midrule
phase gate $S_L$  &
$(i)$ self-Kerr $i^{\hat{n}_2^2}$\\
& $(ii)$ SNAP gate\\
\midrule
$H_L H_M$
& balanced beamsplitter $\pi(H)$ through code deformation \\
$H_L = S_L (H_L H_M) S_L (H_L H_M) S_L$ 
& 
alternate self-Kerr $i^{\hat{n}_2^2}$ and beamsplitter $\pi(H)$\\
\midrule
entangling gate $CZ_{L_1 L_2}$ &
cross-Kerr interaction $(-1)^{\hat{n}_2 \hat{n}_4}$\\
\midrule
$e^{i \theta Z_L Z_M}$ &
quantum Zeno effect with drive $\hat{a}_1^{2} + \hat{a}_1^{\dag 2}$\\
\midrule
$T$-state preparation $\mathcal{P}_{|T\rangle|0\rangle}$ & prepare $|+\rangle$ with Hadamard, followed by quantum Zeno effect\\
\midrule
$T_L$& quantum Zeno effect, with $|\psi\rangle_{M} \in \{|0\rangle, |1\rangle\}$\\
& SNAP gate\\
\bottomrule
\end{tabular}
\end{tcolorbox}
\caption{Possible implementation of a universal set of logical operations for the logical qubit of the two-mode Fourier cat code.}
\label{table:op}
\end{table}

\subsection{Gates in $N(G)$ through code deformation}
\label{sec:N(G)}

An original feature of the encoding occurs when $G \subsetneq N(G)$, \textit{i.e.}~when the normalizer of $G$ (ignoring global phases) is strictly larger than $G$. 
This is the case for instance with the group $G = \langle X, Z\rangle$ since $N(G) = \langle X,Z,H\rangle$ contains the Hadamard gate. 
We recall that the defining representation $\lambda$ and the physical representation $\pi$ extend to the whole unitary group $U(2)$. The idea to implement the logical gate $\lambda(U)$ (that we denote by $U$ for simplicity) is to perform the physical operation $\pi(U)$. In general, this does not leave the codespace invariant. Rather, we show that it induces a logical gate on both encoded qubits and deforms the Fourier code: the final state lives in the Fourier code with encoding map $\mathcal{E}_U$ obtained by replacing the initial coherent state $|\a\rangle$ with $|U\a\rangle$. This is formalized in the following lemma.
\begin{lemma}
For any gate $U \in N(G)$,
\begin{align}
\pi(U) \  \mathcal{E}( |\ell\rangle  |m\rangle) =\mathcal{E}_{U} (U |\ell\rangle \otimes U|m\rangle).
\end{align}
\end{lemma}

\begin{proof}
Consider $U \in N(G)$ and apply $\pi(U)$ to an encoded state
\begin{align}
\pi(U)|\widehat{\ell,m}\rangle &= \sum_{g,h \in G} [\Gamma^{-1/2}]_{g,h}  \langle m |g^\dag |\ell\rangle \pi(U) |h\a\rangle \nonumber\\
&= \sum_{g,h \in G} [\Gamma^{-1/2}]_{g,h}  \langle m | g^\dag |\ell\rangle \pi(UhU^\dag) |U\a\rangle \nonumber\\
&= \sum_{g',h' \in G} [\Gamma^{-1/2}]_{U^\dag g'U,U^\dag h'U}  \langle m |U^\dag {g'}^\dag U |\ell\rangle \pi(h') |U\a\rangle \label{eqn:cov}\\
&= \sum_{g',h' \in G} [\Gamma_U^{-1/2}]_{g',h'}  \langle m | U^\dag{g'}^\dag U |\ell\rangle \pi(h') |U\a\rangle \label{eqn:inv}\\
&= \mathcal{E}_{U} (U|\ell\rangle \otimes U|m\rangle) \nonumber
\end{align}
where we define $[\Gamma_U]_{g,h} = \langle U^\dag g U \a |U^\dag h U \a\rangle= \langle g U \a | h U \a\rangle$ to be the Gram matrix associated with the family $\{ |g U \a\rangle \: : \: g\in G\}$ and $\mathcal{E}_{U}$ is the quantum Fourier encoding using the constellation induced by $|U\a\rangle$ rather than $|\a\rangle$.
The change of variables in \eqref{eqn:cov} is valid since $UgU^\dag, U h U^\dag \in G$ for $U \in N(G)$.
To prove \eqref{eqn:inv}, observe that
\begin{align*}
[\Gamma]_{U^\dag g U, U^\dag h U} &= \langle U^\dag g U \a | U^\dag h U \a\rangle= \langle   g U \a |   h U \a\rangle=  [\Gamma_U]_{g,h}.
\end{align*}
The matrices $\Gamma$ and $\Gamma_U$ are thus related through $\Gamma_U = P \Gamma P^\dag$ for the permutation matrix $P$ with $P_{g,g'}  = \delta_{g', U^\dag g U}$.
Since $\Gamma$ is positive semidefinite and $P$ is a unitary, applying the inverse square root gives $\Gamma_U^{-1/2} = P \Gamma^{-1/2} P^\dag$ and therefore $[\Gamma^{-1/2}]_{U^\dag g U, U^\dag h U} = [\Gamma_U^{-1/2}]_{g,h}$, as needed.
\end{proof}

In the case of the group $\langle X,Z\rangle$, the application of the gate $\pi(H)$, a balanced beamsplitter operation, implements the logical operation $H_L H_M$ and maps the code from $\mathrm{Im} (\mathcal{E})$ to $\mathrm{Im}  (\mathcal{E}_U)$. It is also very useful to get a gate that only applies a Hadamard gate on the logical qubit while leaving the multiplicity qubit alone. 
This is obtained thanks to the gate identity
\begin{align*}
S H S H S = e^{i\pi/4} H
\end{align*}
where $S = \mathrm{diag}(1,i)$ is the phase gate. 
Assuming that the logical gate $S_L$ is available (see Section \ref{sec:S}), a logical $H$ gate that keeps the state in the original encoding $\mathcal{E}$ is therefore obtained through
\begin{align}
H_L = e^{-i\pi/4} S_L (H_L H_M) S_L (H_L H_M) S_L. 
\end{align}
This is because the two factors $H_M$ cancel each other, and because applying the code deformation twice does nothing since $\pi(H) \pi(H) |\a\rangle = |\a\rangle$.

 In general, the new code induced by $\mathcal{E}_U$ may not have the same error correction properties as the initial code, and could perform significantly worse, which would be bad for the overall scheme. For our specific choice of initial state of the form $|\alpha, i \alpha\rangle$, and for $U=H$, it turns out that it does, as discussed in Section \ref{sec:ec}.

\subsection{The logical $S$ and $CZ$ gates}
\label{sec:S}

The most direct way to obtain the Clifford gates $S$ and $CZ$ is similar to the case of rotation-symmetric bosonic codes and relies on the Kerr effect.

\begin{lemma}
The $S$ and $CZ$ gates can be implemented with self-Kerr and cross-Kerr interactions:
\begin{align}
S_L = i^{\hat{n}_2^2}, \qquad CZ_{L_1 L_2} = (-1)^{\hat{n}_2 \hat{n}_4}.
\end{align}
\end{lemma}

\begin{proof}
These properties directly follow from the fact that $\pi(Z) = (-1)^{\hat{n}_2}$ applies a logical $Z$ gate.
For any integer $n$, it holds that $i^{n^2} = \exp\left( i\frac{\pi}{4}(1 - (-1)^n)\right)$ and therefore
\begin{align*}
i^{\hat{n}_2^2} = \exp\left( i\frac{\pi}{4}(\1 - Z_L)\right) = \exp\left( i\frac{\pi}{4}(2 |1\rangle\langle 1|_L \otimes \1_M)\right) = S_L. 
\end{align*} 
Consider two codeword basis states $|\widehat{\ell_1, m_1}\rangle|\widehat{\ell_2, m_2}\rangle$ and apply the controlled-rotation unitary $(-1)^{\hat{n}_2 \hat{n}_4}$:
\begin{align*}
(-1)^{\hat{n}_2 \hat{n}_4}|\widehat{\ell_1, m_1}\rangle|\widehat{\ell_2, m_2}\rangle &= \left((-1)^{\hat{n}_2}\right)^{\hat{n}_4} |\widehat{\ell_1, m_1}\rangle|\widehat{\ell_2, m_2}\rangle\\
&= \left((-1)^{\ell_1} \right)^{\hat{n}_4} |\widehat{\ell_1, m_1}\rangle|\widehat{\ell_2, m_2}\rangle\\
&= \left((-1)^{\hat{n}_4} \right)^{\ell_1}|\widehat{\ell_1, m_1}\rangle|\widehat{\ell_2, m_2}\rangle\\
&= \left((-1)^{\ell_2}\right)^{\ell_1}|\widehat{\ell_1, m_1}\rangle|\widehat{\ell_2, m_2}\rangle\\
&= (-1)^{\ell_1 \ell_2} |\widehat{\ell_1, m_1}\rangle|\widehat{\ell_2, m_2}\rangle
\end{align*}
which again only relies on the fact that $(-1)^{\hat{n}_2} |\widehat{\ell,m} \rangle = (-1)^\ell  |\widehat{\ell,m} \rangle$.
\end{proof}

\subsection{Logical $Z(\theta)$ gates through quantum Zeno effect}
\label{sec:zeno}

A convenient way to perform gates of the form $e^{i \theta Z}$ on a cat qubit is through the quantum Zeno effect~\cite{MLA14}. The idea is to supplement the dissipation that stabilizes the code space with a slow Hamiltonian dynamics, thus leading to a quasi-unitary evolution within the code subspace. For the standard cat qubit, this is achieved with a Hamiltonian of the form $\eps (\hat{a} + \hat{a}^\dag)$. Here, we require a quadratic Hamiltonian $H_\eps = \eps(\hat{a}_1^2 + \hat{a}_1^{\dag 2})$, similar to the 4-legged cat qubit. 

For the value $\alpha = \sqrt{\frac{\pi}{2}}$, the four code-basis states $|\widehat{\ell,m}\rangle$ are products of single-mode cat states that satisfy
\[ \hat{a}_1^2|\widehat{\ell,m}\rangle = (-1)^{\ell + m} \alpha^2  |\widehat{\ell,m}\rangle = \alpha^2 Z_L Z_M|\widehat{\ell,m}\rangle.\]
If $\Pi = \sum_{\ell,m=0}^1 |\widehat{\ell,m}\rangle\langle \widehat{\ell,m}|$ denotes the projector onto the 2-qubit code space, then 
\[ \Pi (\hat{a}^2 + \hat{a}^{\dagger 2}) \Pi = 2 \alpha^2 Z_L Z_M \Pi\]
and therefore $\Pi \exp(i T H_\eps) \Pi$ effectively implements the logical gate $e^{i \theta Z_L Z_M}$ on the code space if $T =  \frac{\theta}{2\alpha^2 \eps}=\frac{\theta}{\pi \eps}$.
In particular, if the state of the multiplicity qubit is $|0\rangle$ (and similarly if it is $|1\rangle$), this corresponds to the logical gate $Z_L(\theta) := e^{ i\theta Z_L}$. 

The main purpose of this gate is to transform a $|+\rangle$ into a $|+_\theta\rangle = \frac{1}{\sqrt{2}}(|0\rangle + e^{i\theta}|1\rangle)$ state. For this, it is crucial to apply it to a state of the form $|\widehat{+,0}\rangle$ since applying it to $|\widehat{0,0}\rangle$ or $|\widehat{+,+}\rangle$ wouldn't yield a useful state. This explains why the gate $H_L H_M$, obtained by applying the beamsplitter $\pi(H)$ as described in Section \ref{sec:N(G)}, was not quite sufficient for our purpose.
Crucially, the state of the multiplicity register should be well-controlled for the gate to perform as expected. This suggests to mainly apply this gate to auxiliary states in order to prepare logical $T$-states for instance, and exploit gate teleportation to implement a logical $T$ gate.

\subsection{Alternative implementation of $Z(\theta)$ with SNAP gates}

A powerful approach to implementing phase gates of the form $e^{i\theta Z_L}$ relies on selective number-dependent arbitrary phase (SNAP) gates~\cite{HVH15}, \cite{RRM20}.
These gates are of the form 
\begin{align}
\mathcal{S}(\vec{\theta}) = \sum_{n=0}^\infty e^{i \theta_n} |n\rangle\langle n|.
\end{align}
It is straightforward to obtain logical phase and $T$ gates by choosing the appropriate sequence of $\theta_n$ on the second bosonic mode:
\begin{align}
S_L = \sum_{n=0}^\infty e^{i \frac{\pi}{2} n^2} |n\rangle \langle n|_2,\qquad T_L = \sum_{n=0}^\infty e^{i \frac{\pi}{4} n^4} |n\rangle\langle n|_2.
\end{align}

\subsection{State preparation and measurements}

As already mentioned, in the special case where $\a = (\sqrt{\pi/2}, i\sqrt{\pi/2})$, the four basis states take a simple product form as in \eqref{eqn:basis}.
Measuring the logical qubit in the $Z$ basis is straightforward by measuring the parity of either mode. 
In fact, it is also possible to measure the operator $Z_L Y_M$ with photon number measurements. 
The eigenstates of $Z_L Y_M$ can be expanded in the Fock basis:
\begin{align*}
|\widehat{0,+i}\rangle & \propto |1_\alpha\rangle |0_{i\alpha}\rangle + i |1_{i\alpha}\rangle |0_\alpha\rangle  \\
& = \sum_{p,q=0}^\infty ((-1)^q - (-1)^p) f_{p,q} |2p+1\rangle |2q\rangle\\
|\widehat{0,-i}\rangle & \propto \sum_{p,q=0}^\infty ((-1)^q + (-1)^p) f_{p,q} |2p+1\rangle |2q\rangle\\
|\widehat{1,+i}\rangle & \propto  \sum_{p,q=0}^\infty ((-1)^q - (-1)^p)f_{p,q}  |2q\rangle|2p+1\rangle\\
|\widehat{1,-i}\rangle & \propto  \sum_{p,q=0}^\infty ((-1)^q + (-1)^p) f_{p,q} |2q\rangle|2p+1\rangle
\end{align*}
for some Poisson-like coefficients $f_{p,q} = \frac{\alpha^{2p+2q+1}}{\sqrt{(2p+1)!(2q)!}}e^{-\alpha^2}$.
Measuring the photon number modulo 4 for each mode distinguishes between the four states, as summarized in Table \ref{table:ZY}.

\begin{table}[htbp]
\centering
\begin{tcolorbox}[colframe=blue!70!black, title=\textbf{Outcomes of the $Z_LY_M$ measurement}, width=0.9\textwidth]
\centering  
\footnotesize
\setlength{\tabcolsep}{12pt} 
\renewcommand{\arraystretch}{1.3} \begin{tabular}{>{\centering\arraybackslash}p{1.2cm} | >{\centering\arraybackslash}p{1.2cm} >{\centering\arraybackslash}p{1.2cm} >{\centering\arraybackslash}p{1.2cm} >{\centering\arraybackslash}p{1.2cm}}
\toprule
$\hat{n}_1 \setminus \hat{n}_2$ & \textbf{0} & \textbf{1} & \textbf{2} & \textbf{3} \\
\midrule
\textbf{0} &  & $|\widehat{1,-i}\rangle$ &  & $|\widehat{1,+i}\rangle$ \\
\midrule
\textbf{1} & $|\widehat{0,-i}\rangle$ &  & $|\widehat{0,+i}\rangle$ &  \\
\midrule
\textbf{2} &  & $|\widehat{1,+i}\rangle$ &  & $|\widehat{1,-i}\rangle$ \\
\midrule
\textbf{3} & $|\widehat{0,+i}\rangle$ &  & $|\widehat{0,-i}\rangle$ &  \\
\bottomrule
\end{tabular}
\end{tcolorbox}
\caption{Possible outcomes of a photon number measurement modulo 4, for both bosonic modes. Note that the parity (outcome modulo 2) is sufficient to recover $Z_L$. If a single photon is lost, one still correctly recovers the value of $Y_M$.}
\label{table:ZY}
\end{table}

\section{Stabilization and error correction}
\label{sec:ec}

The quantum Fourier encoding is valid for any initial two-mode coherent state but the choice of this state impacts the protection offered by the code. Without loss of generality, the state can be taken of the form $|\alpha\rangle |\alpha e^{i \phi}\rangle$ for some $\alpha>0$ and $\phi \in (0,\pi)$ to be optimized. 
Similarly to quantum spherical codes~\cite{JIB23} and to the covariant encoding results from \cite{DL24}, it makes sense to choose parameters that maximize the minimum Euclidean distance between the coherent states in the constellation. This imposes $\phi = \frac{\pi}{2}$. In addition, the choice $\alpha = \sqrt{\frac{\pi}{2}}$ offers several features.
First, for this value, the  restriction of the encoding to $\mathrm{Span}(|\widehat{0,0}\rangle = |1_\alpha\rangle |0_{i\alpha}\rangle, |\widehat{1,0}\rangle = |0_{i \alpha}\rangle |1_\alpha\rangle)$ satisfies the Knill-Laflamme conditions of the pure-loss channel, the dominant source of errors for bosonic codes, at the first order since the 6 states 
\[  |\widehat{0,0}\rangle, \quad |\widehat{1,0}\rangle, \quad \hat{a}_1  |\widehat{0,0}\rangle, \quad \hat{a}_1 |\widehat{1,0}\rangle, \quad \hat{a}_2  |\widehat{0,0}\rangle, \quad \hat{a}_2|\widehat{1,0}\rangle\]
are orthogonal.\footnote{This holds more generally for any state $|\psi\rangle_M$ of the multiplicity register with real amplitudes: $|\psi\rangle = c |0\rangle + s |1\rangle$. Recalling that $\langle k_{i^m \alpha}|\ell_{i^n \alpha}\rangle = \delta_{k,\ell} \delta_{m,n}$ for $k,\ell, m, n \in \{0,1\}$ and the choice $\alpha =\sqrt{\pi/2}$, it is direct to check that the states
\begin{align*}
|\widehat{0,\psi}\rangle &= c |1_\alpha\rangle |0_{i\alpha}\rangle + s|1_{i\alpha}\rangle |0_\alpha\rangle, \qquad
|\widehat{1,\psi}\rangle = c |0_{i\alpha}\rangle |1_\alpha\rangle+ s |0_\alpha\rangle |1_{i\alpha}\rangle,\\
\hat{a}_1 |\widehat{0,\psi}\rangle &\propto c  |0_\alpha\rangle |0_{i\alpha}\rangle + i s|0_{i\alpha}\rangle |0_\alpha\rangle, \qquad
\hat{a}_1 |\widehat{1,\psi}\rangle \propto i c |1_{i\alpha}\rangle |1_\alpha\rangle+ s |1_\alpha\rangle |1_{i\alpha}\rangle, \\
\hat{a}_2 |\widehat{0,\psi}\rangle &\propto ic  |1_\alpha\rangle |1_{i\alpha}\rangle + s|1_{i\alpha}\rangle |1_\alpha\rangle, \qquad
\hat{a}_2 |\widehat{1,\psi}\rangle \propto  c |0_{i\alpha}\rangle |0_\alpha\rangle+ is |0_\alpha\rangle |0_{i\alpha}\rangle
\end{align*}
are all orthogonal. This implies that the Knill-Laflamme conditions are satisfied for the approximate Kraus operators $\{ \1, \hat{a}_1, \hat{a}_2\}$ of the loss channel in the low-loss regime. The true Kraus operators have an additional factor $\sqrt{1-\gamma}^{\ \hat{n}_{1} + \hat{n}_2}$ that maps the coherent state $|g\a\rangle$ of the constellation to $|\sqrt{1-\gamma} g\a\rangle$, but the orthogonality remains approximately satisfied for a loss parameter $\gamma \ll 1$. } 
This is similar to the 4-legged cat qubit and the pair-cat code that can also correct a single photon loss (in either mode for the pair-cat code).
This optimality is confirmed by computing the performance of the Petz-recovery map of this code for the pure-loss channel: see Fig.~\ref{fig:infidelity}. More precisely, the idea to assess the performance of a code against the pure-loss channel is to compute its fidelity of entanglement. 
Starting with a maximally entangled state $|\Phi\rangle_{AR} = \frac{1}{\sqrt{2}} (|00\rangle_{AR} + |11\rangle_{AR})$ between a register $A$ and a reference $R$, one successively encodes system $A$ in the code via the encoding map $\mathcal{E}: |\ell\rangle \mapsto |\widehat{\ell,0}\rangle$, sends it through the pure-loss channel $\mathcal{N}$ and then performs a recovery map $\mathcal{R}$. The entanglement fidelity is then defined as the fidelity $F_{\mathrm{ent}}$ between the output state $(\mathcal{RNE} \otimes \mathrm{id}) (|\Phi\rangle \langle \Phi|)$ and the initial state $|\Phi\rangle$. 
Often, one considers the optimal recovery map and the corresponding fidelity can be computed by solving a semidefinite program~\cite{FSW07}. Here, one takes instead the slightly suboptimal recovery map called the \textit{Petz map} that still offers good performance~\cite{pet88,GLM22} and is easy to compute~\cite{ZHL24}. Details can be found in Appendix \ref{sec:fid}.
Fig.~\ref{fig:infidelity} compares the performance of the 2-mode Fourier code with the other three encodings mentioned previously: the 4-legged single mode cat code, the 2-repetition cat code and the pair-cat code. All these codes, except for the 2-repetition cat code, admit a sweet spot for the value $\alpha$, and perform similarly for this value. By contrast, the 2-repetition cat code gets worse for larger amplitudes. This is expected since losing a single photon leads to an uncorrectable error and the probability of a single-photon loss increases with $\alpha$. 

Finally, the Gram matrix $\Gamma$ becomes diagonal in the Fourier basis exactly for values of $\alpha$ of the form $\sqrt{\frac{p \pi}{2}}$ for some integer $p$. In that case, it means that the encoding $|\ell\rangle \mapsto |\widehat{\ell, \Omega}\rangle$ for an arbitrary state $|\Omega\rangle$ coincides exactly with the covariant encoding of \cite{DL24}. 

\begin{figure*}[htbp]
    \centering
    \includegraphics[width=0.95\linewidth]{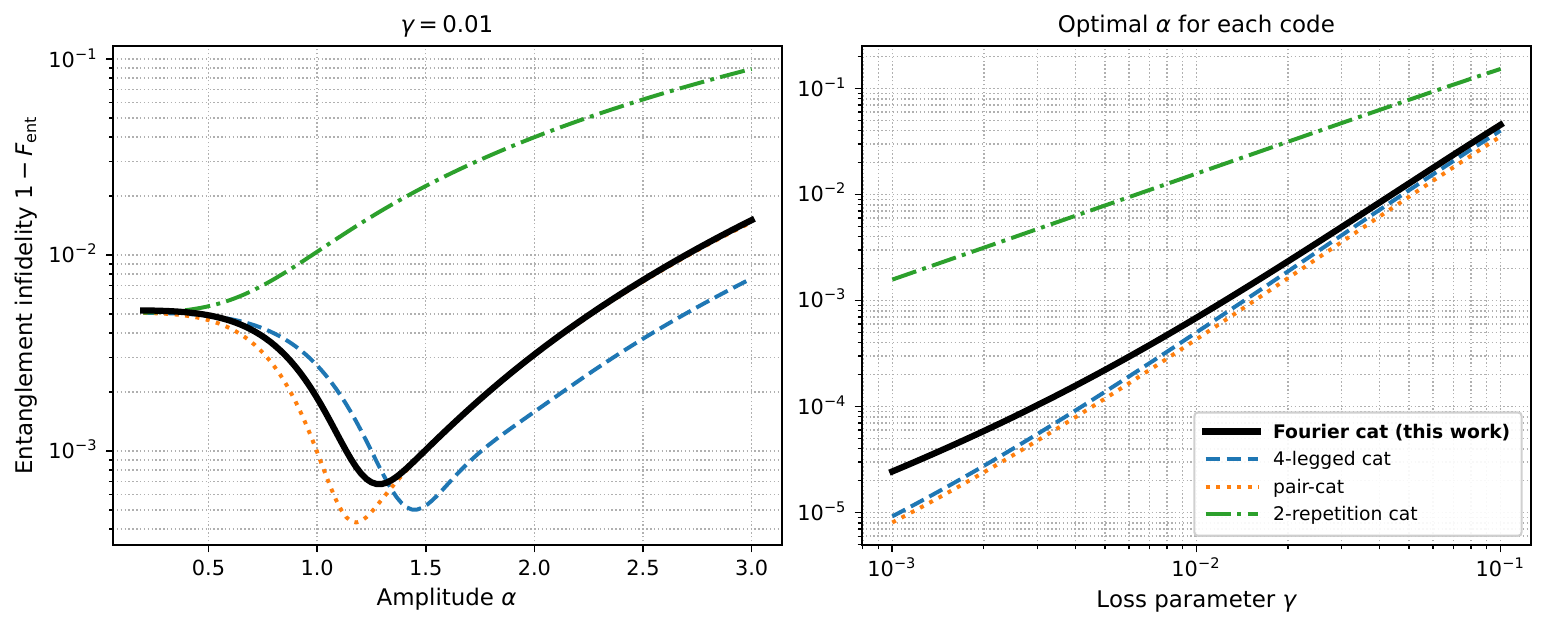}
    \caption{Entanglement infidelity $1-F_{\mathrm{ent}}$ of cat-type bosonic qubit encodings for the pure-loss channel of parameter $\gamma$, using the Petz recovery map. Shown are the two-mode Fourier cat code (logical qubit with the multiplicity register fixed to $|0\rangle_M$), the 4-legged cat code, the 2-repetition cat code, and the pair-cat code. \textbf{Left:} $1-F_{\mathrm{ent}}$ versus coherent amplitude $\alpha$ at fixed $\gamma=0.01$. \textbf{Right:} $1-F_{\mathrm{ent}}$ versus $\gamma$ using $\alpha=\sqrt{\pi/2}$ for the two-mode Fourier cat and 2-repetition cat codes, and $\alpha$ is chosen to minimize the left-panel infidelity at $\gamma=0.01$ for the 4-legged cat and pair-cat codes. Pair-cat values are evaluated in a truncated two-mode Fock basis with cutoffs chosen to ensure numerical convergence. }
    \label{fig:infidelity}
\end{figure*}

Note that the fault-tolerant implementation described in Section \ref{sec:univ} does not visit the full 4-dimensional encoded space. Rather, the idea is to spend most of the computation with a fixed state $|0\rangle_M$ in the multiplicity space, and occasionally $|+\rangle_M$ when one applies the beamsplitter unitary $\pi(H)$ to implement the gate $H_L H_M$ on the deformed code. 
This explains why considering the restriction to $\mathrm{Span}(|\widehat{0,0}\rangle, |\widehat{1,0}\rangle)$ is relevant. Note that when $\pi(H)$ is applied, the code becomes $\mathrm{Span}(|\widehat{0,+}\rangle, |\widehat{1,+}\rangle)$ for an initial coherent state $\pi(H) |\alpha,i\alpha\rangle = |e^{i\pi/4}\alpha, e^{-i\pi/4}\alpha\rangle$ which is isomorphic to the initial code and therefore performs similarly against pure loss.

\subsection{Dissipators and stabilizers}

The full 2-qubit encoding $\mathrm{Span}(|\widehat{\ell,m}\rangle \: : \: \ell, m \in \{0,1\} )$ is a four-dimensional manifold corresponding to the odd-parity sector of the kernel of the following Lindblad operators
\begin{align}\label{eqn:Lindblad}
L_1 = \hat{a}_1^{4} - \alpha^4, \qquad L_2 = \hat{a}_2^{4} - \alpha^4, \qquad L_{12} = \hat{a}_1^2 \hat{a}_2^2 + \alpha^4. 
\end{align}
Note that $L_1$ and $L_2$ are redundant and that one could alternatively only consider the pair of operators $L_1, L_{12}$. Moreover, one could replace $L_1$ or $L_2$ by the operator $L_0 = \hat{a}_1^2 + \hat{a}_2^2$, which has the advantage of being of degree 2 instead of 4. 
The odd parity projector is given by 
\begin{align}
\Pi_{\mathrm{odd}} = \sum_{m + n = 1\mod 2} |m,n\rangle\langle m,n|.
\end{align}
This projector is also relevant for the 2-repetition cat code $\mathrm{Span}(|0_\alpha\rangle|0_\alpha\rangle,|1_\alpha\rangle|1_\alpha\rangle)$ which also imposes a parity constraint on the total number of photons in the 2 bosonic modes. 
The Lindblad operators $L_1$ and $L_2$ are identical to those of the 4-legged cat qubit, and $L_{12}$ coincides with the Lindblad operator of the pair-cat code.
Similarly to the pair code, the parity measurement can be performed without turning off the dissipator $L_{12}$ since it commutes with $e^{i\theta (\hat{n}_1-\hat{n}_2)}$, but this is not the case of $L_1$ or $L_2$ which should be turned off for this parity measurement.

As explained in Section \ref{sec:N(G)}, it is possible to implement a Hadamard gate on both qubits simultaneously by code deformation, simply by applying the beamsplitter gate $\pi(H)$. This sends the initial constellation of coherent states $\{ |g\a\rangle \: : \: g\in \langle X,Z\rangle\}$ to $\{ |gH\a\rangle \: : \: g\in \langle X,Z\rangle\}$. One can check that this constellation lies in the kernel of the Lindblad operators 
\begin{align}
L_1' = \hat{a}_1^{4} + \alpha^4,\qquad L_2' = \hat{a}_2^4+\alpha^4,\qquad \hat{a}_1^2 \hat{a}_2^2 - \alpha^4.
\end{align}

\subsection{Beyond pure loss}

Losses are not the only source of imperfection affecting bosonic systems, and dephasing is another prominent problem. 
One can get a rough idea of the resistance of a bosonic code to dephasing by computing the minimum Euclidean distance between two coherent states in the constellation, and indeed the choice of coherent state $|\alpha, i\alpha\rangle$ is optimal in this respect among states of the form $|\alpha, e^{i\phi} \alpha\rangle$. This only provides some partial information, however, since the superposition is not uniform over the coherent states as in the case of spherical codes~\cite{DL23, JIB23}. 

Evaluating the performance against both losses and dephasing is computationally challenging, especially for a two-mode encoding. 
Nevertheless, it is possible to numerically optimize single-mode bosonic codes that will perform well against a combination of loss and dephasing~\cite{LXJ22}. When the strength of both effects is of the same magnitude, it turns out that the optimal single-mode code is very close  to $\mathrm{Span} (|0_{i\alpha}\rangle, |1_{\alpha}\rangle)$, as shown on Fig. 2 in \cite{LXJ22}. This suggests that the two-mode encoding $|\widehat{0,0}\rangle = |1_{\alpha}\rangle | 0_{i\alpha}\rangle, |\widehat{1,0}\rangle = | 0_{i\alpha}\rangle|1_{\alpha} \rangle$ should be fairly robust against the loss-dephasing channel.

\section{Discussion and future work}
\label{sec:future}

We have focused here on a specific instance of a bosonic quantum Fourier code associated with the Pauli group $D_8 = \langle X,Z\rangle$. Other groups could also be relevant, such as the quaternion group $Q_8 = \langle iX, iZ\rangle$, which has the property that its normalizer $N(Q_8)$ is equal to the binary octahedral group $2O$, \textit{i.e.}~the single-qubit Clifford group. One could therefore implement a gate $S_L S_M$ with a passive Gaussian unitary instead of a self-Kerr evolution. On the other hand, the basis states do not factorize as products of single-mode cat states. It makes sense to investigate other groups, potentially also on qudits of dimension $d$, with encodings over $d$ bosonic modes.

The Gaussian unitary representation $\pi$ is particularly convenient for optical setups, but may be more challenging for circuit-QED architectures. Nevertheless, beamsplitters with excellent fidelities have recently been demonstrated giving hope that the approach outlined here could become realistic in the coming years~\cite{LMG23,CGX23}. The stabilization of the code space with dissipation requires Lindblad operators of degree 4, similarly to the 4-legged cat qubit and the pair-cat codes, which could prove challenging. 

While we have discussed some basic properties of the error correcting code, we have not explicitly assessed how much the suggested implementation of the logical gates would be impacted by experimental imperfections. It is also crucial to understand how errors would propagate in the circuit. While not a full proof, the fact that many physical gates described in Table \ref{table:op} have also been studied for the main bosonic codes out there is a source of optimism concerning their potential fault tolerance. 
More importantly, this table only provides a list of possible physical implementations for the logical operations, but better alternatives likely exist.

Can the Fourier encoding be extended beyond bosonic codes? It is tempting to consider multiqubit codes with the tensor representation $\pi(g) = g^{\otimes n}$. This approach doesn't work directly, however, because the states $\pi(g) |\Phi\rangle$ will not be linearly independent, for any choice of initial state $|\Phi\rangle$, as soon as the group $G$ contains nontrivial phases such as $-1$. Spin encodings and rotors are natural candidates to explore beyond bosonic codes.

\acknowledgements{I thank Aur\'elie Denys, Mazyar Mirrahimi, Cl\'ement Poirson and Christophe Vuillot for many useful discussions about bosonic codes, and Lev-Arcady Sellem for comments on an earlier version of the manuscript.}



\appendix


\section{Single-mode cat qudit}
\label{sec:cat}

One recovers single-mode $d$-dimensional cat encodings by considering the abelian group $G = \mathbbm{Z}/N \mathbbm{Z}$ for $N = dM$, a multiple of $d$, together with its representations $\lambda_k(1) = \omega^k$ and infinite-dimensional representation $\pi(1) = \omega^{\hat{n}}$ for $\omega = e^{2\pi i/N}$.
The Fourier transform over $G$ is 
\begin{align}
F = \frac{1}{\sqrt{N}} \sum_{k,\ell=0}^{N-1} \omega^{k\ell} |\lambda_k\rangle \langle \ell |.
\end{align} 
Fix $\alpha>0$. The Gram matrix $\Gamma$ associated with the family $\pi(k)|\alpha\rangle$ is $\Gamma_{k,\ell} = e^{\alpha^2(-1+\omega^{\ell-k})}$
and is diagonalized by the Fourier transform $\Gamma = F \mathrm{diag}(\Delta_0, \Delta_1, \ldots, \Delta_{N-1}) F^\dag$ with 
\[ \Delta_k = e^{-\alpha^2} \sum_{\ell=0}^{N-1} \omega^{k\ell} e^{\alpha^2 \omega^\ell}.\]
One obtains an orthonormal basis using $|\ell\rangle := \Gamma^{-1/2} |\omega^\ell \alpha\rangle$:
\[ |\ell\rangle = \frac{1}{N} \sum_{k,p=0}^{N-1} \omega^{p(\ell-k)} \Delta_p^{-1/2}  |\omega^k \alpha\rangle.\]
For the two-mode cat encoding associated to the noncommutative group $G = \langle X,Z\rangle$, one can encode a qubit in a 2-dimensional representation of the group. Here, the group is abelian and therefore all the irreducible representations are 1-dimensional. The natural strategy is therefore to encode the basis elements in distant irreducible representations through $\mathrm{Enc}(|k\rangle) := F^\dag |\lambda_{kM}\rangle$ for $0 \leq k <d$. 
This corresponds to the standard cat code:
\begin{align}
\mathrm{Enc}(|k\rangle) = \frac{1}{\sqrt{N \Delta_{kM}}} \sum_{p=0}^{dM-1}  \omega^{-kpM}  |\omega^p \alpha\rangle.
\end{align}

\section{Fidelity of entanglement}
\label{sec:fid}

Consider a pure-loss channel of parameter $t$. It can be represented as a unitary operator $U$ acting on the input mode and on an additional mode in the vacuum state, that maps an initial coherent state $|\beta\rangle$ to
$$ U |\beta\rangle |0\rangle = |t\beta \rangle_t |r\beta\rangle_r$$
where $r = \sqrt{1-t^2}$, and where we label by $t$ and $r$ the two output modes. We also use $\gamma := 1-t^2$ so that $t=\sqrt{1-\gamma}$ and $r=\sqrt{\gamma}$.
One can choose the following Kraus operators for this channel:
$$C_p = \langle p|_r U \Pi \otimes |0\rangle$$
where we define the orthonormal basis $|p\rangle_r$ for the second output mode as $|p\rangle_r = \sum_{q\in G} [\Gamma_r^{-1/2}]_{q,p} |rq\a\rangle$ with the relevant Gram matrix $[\Gamma_{r}]_{g,h} = \langle rg\a|rh\a\rangle$. 
These are indeed Kraus operators since they satisfy the usual normalization condition
\begin{align*}
\sum_{p\in G} C_p^\dag C_p&= \sum_{p\in G} (\Pi \otimes \langle 0 | ) U^\dag (\mathbbm{1}_t \otimes |p\rangle \langle p|_r) U (\Pi \otimes |0\rangle\\
&=  (\Pi \otimes \langle 0|) U^\dag (\mathbbm{1}_t \otimes \1_r) U (\Pi \otimes |0\rangle\\
&=  (\Pi \otimes \langle 0| )  (\Pi \otimes |0\rangle\\
&= \Pi 
\end{align*}
and they act as desired on coherent states:
\begin{align*}
\sum_{p\in G} C_p |g\a\rangle \langle h\a| C_p^\dag&=\sum_{p\in G} \langle p|_r U (\Pi \otimes |0\rangle ) |g\a\rangle \langle h\a| (\Pi \otimes \langle 0 | ) U^\dag  |p\rangle_r\\
&=\sum_{p\in G} \langle p|_r U  |g\a\rangle |0\rangle   \langle h\a|\langle 0|  U^\dag |p\rangle_r\\
&=\sum_{p\in G} \langle p|_r   |tg\a\rangle |rg\a\rangle   \langle th\a|\langle rh\a|   p\rangle_r\\
&=|tg\a\rangle   \langle th\a| \sum_{p\in G}  \langle rh\a|   p\rangle \langle p|_r   |rg\a\rangle \\
&=|tg\a\rangle   \langle th\a| \langle rh\a    |rg\a\rangle 
\end{align*}
In order to compute the fidelity of entanglement for this channel for the Petz recovery map, one should first evaluate the QEC matrix $M$ with entries $M_{[k,p],[\ell, q]} := \langle \widehat{k,0} |C_p^\dag C_q |\widehat{\ell,0} \rangle$. The fidelity $F_\mathrm{ent}$ can then be conveniently computed as~\cite{ZHL24}
\begin{align}\label{eqn:fid-petz}
F_\mathrm{ent} & = \frac{1}{d^2}  \left\| \mathrm{tr}_{L} M^{1/2} \right\|_{\mathrm{hs}}^2,
\end{align}
where $\| \cdot\|_{\mathrm{hs}}$ stands for the Hilbert-Schmidt norm.
Computing $M$ is straightforward:
\begin{align*}
C_p |\widehat{k,0} \rangle &= C_p \sum_{g\in G} [\Gamma^{-1/2} F^\dag]_{g,k0} |g\a\rangle\\
&=  \sum_{g\in G} [\Gamma^{-1/2} F^\dag]_{g,k0} |tg\a\rangle \langle p|rg\a\rangle\\
&=  \sum_{g\in G} [\Gamma^{-1/2} F^\dag]_{g,k0}  [\Gamma_r^{1/2}]_{p,g}  |tg\a\rangle
\end{align*}
and therefore
\begin{align*}
M_{[k,p],[\ell, q]} &= \langle \widehat{k,0} |C_p^\dag C_q |\widehat{\ell,0} \rangle\\
&= \sum_{g,h \in G}[F \Gamma^{-1/2} ]_{k0, g} [\Gamma^{-1/2} F^\dag]_{h,\ell 0}  [\Gamma_r^{1/2}]_{gp}   [\Gamma_r^{1/2}]_{qh} [\Gamma_t]_{gh}.
\end{align*}
The entries of the three Gram matrices $\Gamma, \Gamma_t, \Gamma_r$ are simple functions of the $|G|$-dimensional Gram matrix of the 2-dimensional family of vectors $\lambda(g)|+_i\rangle \in \mathbbm{C}^2$, with $|+_i\rangle = \frac{1}{\sqrt{2}} (|0\rangle + i|1\rangle)$:
\begin{align}\label{eqn:3Gram}
[\Gamma]_{gh} = e^{2 \alpha^2(\Lambda_{gh}-1) },\qquad   [\Gamma_t]_{gh} = e^{2 (1-\gamma) \alpha^2(\Lambda_{gh}-1) } ,\qquad  [\Gamma_r]_{gh} = e^{2 \gamma \alpha^2(\Lambda_{gh}-1) }.
\end{align}

\end{document}